\documentclass[11pt]{article}

\usepackage{microtype}
\usepackage[margin=1in]{geometry}
\usepackage{amsmath,array,bm}
\usepackage{amsmath,amsfonts}
\usepackage{amsthm}
\usepackage{amssymb,latexsym}
\usepackage{algorithm}
\usepackage{subcaption}
\usepackage{algorithmicx}
\usepackage[]{algpseudocode}
\usepackage[dvipsnames]{xcolor}
\usepackage{graphicx}
\usepackage{caption}
\usepackage{enumitem}
\usepackage{float}
\usepackage{thmtools}
\usepackage{hyperref}
\usepackage{mathtools}
\usepackage{cleveref}
\usepackage{bussproofs}
\usepackage[square,sort,comma,numbers]{natbib}
\usepackage{tikz}
\usepackage{xspace}
\usepackage{xcolor}
\usepackage{tipa}
\usepackage{mdframed}
\usepackage[export]{adjustbox}[2011/08/13]
\usepackage{multirow}
\usepackage{hhline}

\newtheorem{theorem}{Theorem}
\newtheorem{claim}{Claim}

\theoremstyle{definition}

\crefname{invar}{invariant}{invariants}
\crefname{ineq}{inequality}{inequalities}
\crefname{constr}{constraint}{constraints}
\crefname{tbl}{table}{tables}
\crefname{lem}{lemma}{lemmata}
\crefname{lemma}{lemma}{lemmata}
\crefname{cond}{condition}{conditions}

\title{On the Partition Set Cover Problem}
\author{Tanmay Inamdar \\ \texttt{tanmay-inamdar@uiowa.edu} \and Kasturi Varadarajan \\\texttt{kasturi-varadarajan@uiowa.edu}}
\date{Department of Computer Science\\The University of Iowa}

\newcommand{\real}{\mathbb{R}}

\newcommand{\X}{X}

\newcommand{\R}{\mathcal{R}}
\newcommand{\B}{\mathcal{B}}

\ifdefined\C
\renewcommand{\C}{\mathcal{C}}
\else
\newcommand{\C}{\mathcal{C}}
\fi
\renewcommand{\L}{\mathcal{L}}
\newcommand{\M}{\mathcal{M}}

\renewcommand{\deg}{\textsf{deg}}
\newcommand{\A}{\mathcal{A}}

\newcommand{\E}{\operatorname{E}}
\newcommand{\Var}{\mathrm{Var}}

\renewcommand{\tilde}{\widetilde}
\renewcommand{\S}{\mathcal{R}}

\newcommand{\K}{\mathcal{K}}
\newcommand{\PSC}{\textsf{PSC}\xspace}

\newcommand{\lr}[1]{\left(#1\right)}

\renewcommand{\epsilon}{\varepsilon}

\newcommand{\LP}{\textsf{LP}\xspace}

\newcommand{\hi}{\ensuremath{\langle S_i, \Sigma'_\ell \rangle}\xspace}

\newcommand{\NP}{\textsf{NP}}

\begin{document}
	\maketitle
	\begin{abstract}
		Several algorithms with an approximation guarantee of $O(\log n)$ are known for the Set Cover problem, where $n$ is the number of elements. We study a generalization of the Set Cover problem, called the Partition Set Cover problem. Here, the elements are partitioned into $r$ \emph{color classes}, and we are required to cover at least $k_t$ elements from each color class $\C_t$, using the minimum number of sets. We give a randomized \LP-rounding algorithm that is an $O(\beta + \log r)$ approximation for the Partition Set Cover problem. Here $\beta$ denotes the approximation guarantee for a related Set Cover instance obtained by rounding the standard \LP. As a corollary, we obtain improved approximation guarantees for various set systems for which $\beta$ is known to be sublogarithmic in $n$. We also extend the \LP rounding algorithm to obtain $O(\log r)$ approximations for similar generalizations of the Facility Location type problems. Finally, we show that many of these results are essentially tight, by showing that it is \NP-hard to obtain an $o(\log r)$-approximation for any of these problems.
	\end{abstract}
	
	\section{Introduction}
	We first consider the Set Cover problem. The input of the Set Cover problem consists of a set system $(X, \R)$, where $X$ is a set of $n$ elements and $\R$ is a collection of subsets of $X$. Each set $S_i \in \R$ has a non-negative weight $w_i$. The goal of the set cover problem is to find a minimum-weight sub-collection of sets from $\R$, that \emph{covers} $X$. In the unweighted version, all weights are assumed to be $1$. We state the standard \LP relaxation for the Set Cover problem.
	\begin{mdframed}[backgroundcolor=gray!9] 
		(Set Cover \LP)
		\begin{alignat}{3}
		\text{minimize}   \displaystyle&\sum\limits_{S_i \in \S} w_{i}x_{i} & \nonumber \\
		\text{subject to} \displaystyle&\sum\limits_{i:e_{j} \in S_{i}}   x_{i} \geq 1,  \quad &  \forall e_j \in \X \label[constr]{constr:sc-cover-ej}\\
		\displaystyle &x_i \in [0, 1], &  \forall S_i \in \S \label[constr]{constr:sc-fractional-x}
		\end{alignat}
	\end{mdframed}

	It is well-known that a simple greedy algorithm, or an \LP rounding algorithm gives an $O(\log n)$ approximation. This can be improved to $O(\log \Delta)$, where $\Delta$ is the maximum size of a set -- see \cite{VaziraniBook} for further details and references. It is also known that it is not possible to obtain an approximation guarantee that is asymptotically smaller than $O(\log n)$ in general, under certain standard complexity theoretic assumptions (\cite{Feige1998,DS2014}). A simple \LP rounding algorithm is also known to give an $f$ approximation (see \cite{VaziraniBook}), where $f$ is the maximum \emph{frequency} of an element, i.e., the maximum number of sets any element is contained in. For several set systems such as geometric set systems, however, sublogarithmic\,---\,or even constant\,---\,approximation guarantees are known. For a detailed discussion of such results, see \cite{Inamdar2018partial}. 
	
	Partial Set Cover problem (\PSC) is a generalization of the Set Cover problem. Here, along with the set system $(X, \R)$, we are also given the coverage requirement $1 \le k \le n$. The objective of \PSC is to find a minimum-weight cover for at least $k$ of the given $n$ elements. It is easy to see that when the coverage requirement $k$ equals $n$, \PSC reduces to the Set Cover problem. However, for a general $k$, \PSC introduces the additional difficulty of discovering a subset with $k$ of the $n$ elements that one must aim to cover in order to obtain a least expensive solution. Despite this, approximation guarantees matching that for the standard Set Cover are known in many cases. For example, a slight modification of the greedy algorithm can be shown to be an $O(\log \Delta)$ approximation (\cite{Kearns1990,Slavik1997}). Algorithms achieving the approximation guarantee of $f$ are known via various techniques -- see \cite{Fujito04} and the references therein. For several instances of \PSC, an $O(\beta)$ approximation algorithm was described in \cite{Inamdar2018partial}, where $\beta$ is the integrality gap of the standard set cover \LP for a related set system. As a corollary, they give improved approximation guarantees for the geometric instances of \PSC, for which $\beta$ is known to be sublogarithmic (or even a constant).
	
	\paragraph{Partition Set Cover Problem.}
	Now we consider a further generalization of the Partial Set Cover problem, called the Partition Set Cover problem. Again, the input contains a set system $(X, \R)$, with weights on the sets, where $X = \{e_1, \ldots, e_n\}$ and $\R = \{S_1, \ldots, S_m\}$. We are also given $r$ non-empty subsets of $X$: $\C_1, \ldots, \C_r$, where each $\C_t$ is referred to as a \emph{color class}. These $r$ color classes form a partition of $X$. Each color class $\C_t$ has a coverage requirement $1 \le k_t \le |\C_t|$ that is also a part of the input. The objective of the Partition Set Cover problem is to find a minimum-weight sub-collection $\R' \subseteq \R$, such that it meets the coverage requirement of each color class, i.e., for each color class $\C_t$, we have that $|(\bigcup \R') \cap \C_t| \ge k_t$. Here, for any $\R' \subseteq \R$, we use the shorthand $\bigcup \R'$ for referring to the union of all sets in $\R'$, i.e.,  $\bigcup \R' \coloneqq \bigcup_{S_i \in \R'} S_i$.
	
	\citet{bera2014approximation} give an $O(\log r)$-approximation for an analogous version of the Vertex Cover problem, called the Partition Vertex Cover problem. This is a special case of the Partition Set Cover problem where each element is contained in exactly two sets. The Vertex Cover and Partial Vertex Cover problems are, respectively, special cases of the Set Cover and Partial Set Cover problems. 
        Various $2$ approximations are known for Vertex Cover (see references in \cite{VaziraniBook}) as well as Partial Vertex Cover (\cite{BshoutyL1998,Hochbaum98,Gandhi2004}).
        \citet{bera2014approximation} note that for the Partition Set Cover problem, an extension of the greedy algorithm of \citet{Slavik1997} gives an $O(\log (\sum_{t = 1}^r k_t))$ approximation. On the negative side, they (\cite{bera2014approximation}) show that it is \NP-hard to obtain an approximation guarantee asymptotically better than $O(\log r)$ for the Partition Vertex Cover problem. Since Partition Vertex Cover is a special case of Partition Set Cover, the same hardness result holds for the Partition Set Cover problem as well. 
	
	\citet{Har2018few} consider a problem concerned with breaking up a collection of point sets using a minimum number of hyperplanes. They reduce this problem to an instance of the Partition Set Cover problem, where the color classes are no longer required to form a partition of $X$. For this problem, which they call `Partial Cover for Multiple Sets', they describe an $O(\log (nr))$ approximation. We note that the algorithm of \citet{bera2014approximation} as well as our algorithm easily extends to this more general setting.

	\subsection{Natural \LP Relaxation and Its Integrality Gap} \label{subsec:nat-lp}
	
	Given the success of algorithms based on the natural \LP relaxation for \PSC, let us first consider the natural \LP for the Partition Set Cover problem. We first state this natural \LP relaxation.
	
	\begin{mdframed}[backgroundcolor=gray!9] 
		(Natural \LP)
		\begin{alignat}{3}
		\text{minimize}   \displaystyle&\sum\limits_{S_i \in \S} w_{i}x_{i} & \nonumber \\
		\text{subject to} \displaystyle&\sum\limits_{i:e_{j} \in S_{i}}   x_{i} \geq z_j,  \quad &  \forall e_j \in \X \label[constr]{constr:cover-ej}\\
		\displaystyle&\sum_{e_j \in \C_r}z_j \ge k_t, & \forall \C_t \in \{\C_1, \ldots, \C_r\} \label[constr]{constr:cover-ci}\\
		\displaystyle &z_j \in [0, 1], &  \forall e_j \in \X \label[constr]{constr:fractional-z}\\
		\displaystyle &x_i \in [0, 1], &  \forall S_i \in \S \label[constr]{constr:fractional-x}
		\end{alignat}
	\end{mdframed}
	
	In the corresponding integer program, the variable $x_i$ denotes whether or not the set $S_i \in \R$ is included in the solution. Similarly, the variable $z_j$ denotes whether or not an element $e_j \in X$ is covered in the solution. Both types of variables are restricted to $\{0, 1\}$ in the integer program. The \LP relaxation stated above relaxes this condition by allowing those variables to take any value from $[0, 1]$.

	One important way the Standard \LP differs from the Partial Cover \LP is that we have a coverage constraint (\Cref{constr:cover-ci}) for each color class $\C_1, \ldots, \C_r$. Unfortunately, this \LP has a large integrality gap as demonstrated by the following simple construction.
	
	\paragraph{Integrality Gap.}
	Let $(X, \R)$ be the given set system, where $X = \{e_1, e_2, \ldots, e_n\}$ and $\R = \{S_1, S_2, \ldots, S_{\sqrt{n}}\}$ -- assuming $n$ is a perfect square. The sets $S_i$ form a partition of $X$, such that each $S_i$ contains exactly $\sqrt{n}$ elements. For any $1 \le i \le \sqrt{n}$, the color class $\C_i$ equals $S_i$, and its coverage requirement, $k_i = 1$. Also, for each set $S_i$, $w_i = 1$. Clearly any integral optimal solution must choose all sets $S_1, \ldots, S_{\sqrt{n}}$, with cost $\sqrt{n}$.
	
	On the other hand, consider a fractional solution $(x, z)$, where for any $S_i \in \R$, $x_i = \frac{1}{\sqrt{n}}$; and for any $e_j \in X$, $z_j = \frac{1}{\sqrt{n}}$. It is easy to see that this solution satisfies all constraints, and has cost $1$. We emphasize here that even the natural \PSC \LP has a large integrality gap, but it can be easily circumvented via parametric search; see \cite{Gandhi2004,Inamdar2018partial} for examples. For the Partition Set Cover problem, however, similar techniques do not seem to work.
	
	\subsection{Our Results and Techniques}
	For any subset $Y \subseteq X$, let $(Y, \R_{|Y})$ denote the projection of $(X, \R)$ on $Y$, where $\R_{|Y} = \{S_i \cap Y \mid S_i \in \R \}$. Suppose there exists an algorithm that can round a feasible \emph{Set Cover} \LP solution for any projection $(Y, \R_{|Y})$, within a factor of $\beta$. Then, we show that there exists an $O(\beta + \log r)$ approximation for the Partition Set Cover problem on the original set system $(X, \R)$. 
	
	Given the integrality gap of the natural \LP for the Partition Set Cover problem, we strengthen it by adding the knapsack cover inequalities (first introduced by \citet{carrStrengthening}) to the \LP relaxation -- the details are given in the following section. This approach is similar to that used for the Partition Vertex Cover problem considered in \citet{bera2014approximation}, and for the Partial Set Cover problem in \citet{Fujito04}. A similar technique was also used for a scheduling with outliers problem by \citet{gupta2009scheduling}.
	
	Once we have a solution to this strengthened \LP, we partition the elements of $X$ as follows. The elements that are covered to an extent of at least a positive constant in this \LP solution are said to be \emph{heavy} (the precise definition is given in the following section), and the rest of the elements are \emph{light}. For the heavy elements, we obtain a standard \emph{Set Cover} \LP solution, and round it using a black-box rounding algorithm with a guarantee $\beta$. 
	
	To meet the residual coverage requirements of the color classes, we use a randomized algorithm that covers some of the light elements. This randomized rounding consists of $O(\log r)$ independent iterations of a simple randomized rounding process. Let $\Sigma_\ell$ be a random collection of sets, which is the outcome of an iteration $\ell$. The cover for the heavy elements from previous step, plus $\bigcup_\ell \Sigma_\ell$, which is the result of randomized rounding, make up the final output of our algorithm. The analysis of this solution hinges on showing the following two properties of $\Sigma_\ell$. 
	\begin{enumerate}
		\item Expected cost of $\Sigma_\ell$ is no more than some constant times the optimal cost, and
		\item For any color class $\C_t$, the sets in $\Sigma_\ell$ satisfy the residual coverage of $\C_t$ with at least a constant probability.
	\end{enumerate}
	The first property of $\Sigma_\ell$ follows easily, given the description of the randomized rounding process. Much of the analysis is devoted to showing the second property. The randomized rounding algorithm ensures that in each iteration, the residual requirement of each color class is satisfied in expectation. However, showing the second property, i.e., the residual coverage is satisfied with at least a constant probability, is rather tricky for the Partition Set Cover problem. The analogous claim about the Partition Vertex Cover problem in \cite{bera2014approximation} is easier to prove, because each edge is incident on exactly two vertices. Despite the fact that here, an element can belong to any number of sets, we are able to show that the second property holds, using a careful analysis of the randomized rounding process.
	
	From these two properties, it is straightforward to show that the algorithm is an $O(\beta + \log r)$ approximation for the Partition Set Cover problem (for details, see \Cref{thm:main-theorem}). As shown in \cite{bera2014approximation}, $\Omega(\log r)$ is necessary, even for the Partition Vertex Cover problem, and we extend this result to the Partition Set Cover problem induced by various geometric set systems (see \Cref{sec:hardness}). Also $\beta$ is the integrality gap of the standard Set Cover \LP for this set system. Our approximation guarantee of $O(\beta + \log r)$ should be viewed in the context of these facts. We also note here that an extension of the \LP rounding algorithm of \cite{Inamdar2018partial} for the \PSC can be shown to be an $O(\beta + r)$ approximation. In particular, this extension involves doing a separate rounding for the light elements in each color class, which does not take advantage of the fact that a set in $\R$ may cover elements from multiple color classes. When $\beta$ is a constant, our guarantee is an exponential improvement over $O(\beta + r)$.

        Our analysis of the randomized rounding process may be of independent interest, and at its core establishes the following type of claim. Suppose we have a set system $(X, \R)$, where each set $S_i \in \R$ has at most $k$ elements; here $k$ is a parameter that can be much smaller than $|X|$. Suppose that we construct a collection of subsets $\Sigma \subseteq \R$ by independently picking each set $S_i \in \R$ with a certain probability $0 < \mu_i < 1$. Assume that the set system and the $\mu_i$ ensure the following properties:
        \begin{enumerate}
        \item The expected number of elements that $\Sigma$ covers is large in terms of $k$, say at least $6k$.
        \item For any element in $X$, the probability that it is covered by $\Sigma$ is at most a constant strictly smaller than $1$.
        \end{enumerate}
        Given these conditions, our analysis shows that $\Sigma$ covers at least $k$ elements with probability at least a positive constant.
	
	\paragraph{Applications.}
	
	As described earlier, for several geometric set systems, the approximation guarantee $\beta$ for Set Cover based on the natural LP relaxation is sublogarithmic. For example, when $X$ is a set of points in $\real^2$ (resp. $\real^3$), and each set in $\R$ consists of the points in $X$ contained in some input disk (resp. halfspace), an \LP rounding algorithm with a guarantee $\beta = O(1)$ is known. Therefore, for the corresponding Partition Cover instances, we get an $O(\log r)$ approximation. If $X$ is a set of $n$ rectangles in $\real^3$, and each set in $\R$ is the subset of rectangles that are stabbed by an input point, then an \LP rounding algorithm with a guarantee $\beta = O(\log \log n)$ is known. In the corresponding Partition Set Cover instance, we get an $O(\log \log n + \log r)$ approximation. We summarize some of these results in the following table. We remind the reader that it is \NP-hard to obtain an $o(\log r)$ approximation for all of these set systems, as shown in \Cref{sec:hardness}. Therefore, when $\beta = O(\log r)$, improving on the $O(\beta + \log r)$ approximation is \NP-hard. Otherwise, when $\beta = \omega(\log r)$, improving $O(\beta + \log r)$ also involves improving the approximation guarantee for the corresponding Set Cover problem, for example, hitting rectangles in $\real^3$, or covering by fat triangles in $\real^2$.
	
	
	\begin{table}[H]
		\centering
		\caption{Some of the approximation guarantees for Partition Set Cover (last column). In the third column, we have $\beta$, the \LP-based approximation guarantee for the corresponding Set Cover instance. See \cite{ClarksonV2007,AronovES2010,VaradarajanWGSC2010, EASFat, ChanGKS12,ElbTerrain} for the references establishing these bounds on $\beta$.}
		\begin{tabular}{|c|c|c|c|}
			\hline
			$\X$ & Geometric objects inducing $\S$ & $\beta$ & Our guarantee\\
			\hline
			\multirow{2}{*}{Points in $\real^2$} & Disks (via containment) & $O(1)$ & $O(\log r)$ \\
			\hhline{|~|---}
			& Fat triangles (containment) & $O(\log \log^* n)$ & $O(\log \log^* n + \log r)$ \\
			\hline
			\multirow{2}{*}{Points in $\real^3$} & Unit cubes (containment) & $O(1)$ & $O(\log r)$ \\
			\hhline{|~|---}
			& Halfspaces (containment) & $O(1)$ & $O(\log r)$ \\
			\hline
			Rectangles in $\real^3$ & Points (via stabbing) & $O(\log \log n)$ & $O(\log \log n + \log r)$ \\
			\hline
			Points on 1.5D terrain & Points on terrain (via visibility) & $O(1)$ & $O(\log r)$  \\
			\hline
		\end{tabular}
	\end{table}
	
	As a combinatorial application, consider the combinatorial version of the Partition Set Cover problem, where each element is contained in at most $f$ sets of $\R$. We note that the algorithm of \citet{bera2014approximation} can be extended in a straightforward way to obtain an $O(f \log r)$ approximation in this case. However, recall that the Set Cover \LP can be rounded to give an $f$ approximation. Therefore, using our result, we can get an $O(f + \log r)$ approximation for the Partition Set Cover problem, which is an improvement over the earlier result. 
	
	Finally, in \Cref{sec:fl-mcc}, we consider analogous generalizations of the (Metric Uncapacitated) Facility Location Problem, and the so-called Minimum Cost Covering Problem (\cite{CharikarP04}). Various \LP-based $O(1)$ approximations are known for these problems (\cite{JainVazirani2001,ByrkaFLLP,Li2013}, and \cite{CharikarP04} respectively), i.e., for these problems, $\beta = O(1)$. We show how to adapt the algorithm from \Cref{sec:LPRounding} to obtain $O(\log r)$ approximations for the generalizations of these problems with $r$ color classes.
	
	\paragraph{Organization.}
	In the following section, we first describe the strengthened \LP and then discuss the randomized rounding algorithm. In \Cref{subsec:coverage}, we prove the second property of $\Sigma_\ell$ as stated above. In \Cref{subsec:expectation} we prove a technical lemma that is required in the \Cref{subsec:coverage}. In \Cref{subsec:ellipsoid}, we show how to obtain a feasible solution to the strengthened \LP, despite exponentially many constraints in the \LP. In \Cref{sec:fl-mcc}, we give the $O(\log r)$ approximations for the generalizations of the Facility Location and the Minimum Cost Covering problems. Finally, in \Cref{sec:hardness}, we show how to extend the $\Omega(\log r)$ hardness result from \cite{bera2014approximation}, for all these problems.

	\section{Strenghtened \LP and Randomized Rounding} \label{sec:LPRounding}
	Recall that the input contains a set system $(X, \R)$, with weights on the sets, where $X = \{e_1, \ldots, e_n\}$ and $\R = \{S_1, \ldots, S_m\}$. We are also given $r$ non-empty subsets of $X$: $\C_1, \ldots, \C_r$, where each $\C_t$ is referred to as a \emph{color class}. Here, we consider a generalization of the Partition Set Cover problem, where the color classes cover $X$ but are no longer required to form a partition of $X$, i.e., an element $e_j$ can belong to multiple color classes. Each color class $\C_t$ has a coverage requirement $1 \le k_t \le |\C_t|$. The objective of the Partition Set Cover problem is to find a minimum-weight sub-collection $\R' \subseteq \R$, such that it meets the coverage requirement of each color class, i.e., for each color class $\C_t$, we have that $|(\bigcup \R') \cap \C_t| \ge k_t$. Let $OPT$ denote the cost of an optimal solution for this problem.
	
	Now, we describe the strengthened \LP. Imagine $\A \subseteq \R$ is a collection of sets that we have decided to add to our solution. The sets in $\A$ may cover some elements from each color class $\C_t$, potentially reducing the remaining coverage requirement. For a color class $\C_t$, define $\C_t(\A) \coloneqq (\bigcup \A) \cap \C_t$ to be the set of elements of color $\C_t$, covered by the sets in $\A$. Then, let $k_t(\A) \coloneqq \max\{0, k_t - |\C_t(\A)| \}$ be the residual coverage requirement of $\C_t$ with respect to the collection $\A$. Finally, for a set $S_i \not\in \A$, and a color class $\C_t$, define $\deg_t(S_i, \A) \coloneqq |S_i \cap (\C_t \setminus \C_t(\A))|$ to be the additional number of elements from $\C_t$ covered by $S_i$, provided that $\A$ is already a part of the solution. For any $\C_t$ and for any collection $\A \subseteq \R$, the following constraint is satisfied by any feasible integral solution:
	$$\sum_{S_i \not\in \A} x_i \cdot \min\{\deg_t(S_i, \A), k_t(\A)\} \ge k_t(\A)$$
	
	For a detailed explanation of why this constraint holds, we refer the reader to the discussion in \citet{bera2014approximation}. Adding such constraints for each $\A \subseteq \R$ and for each $\C_t$, gives the following strengthened \LP.
	
	\begin{mdframed}[backgroundcolor=gray!9] 
		(Strengthened \LP)
		\begin{alignat}{3}
		\text{minimize}   \displaystyle&\sum\limits_{S_i \in \S} w_{i}x_{i} & \nonumber \\
		\text{subject to\quad } 
		&\text{\Cref{constr:cover-ej,constr:cover-ci,constr:fractional-z,constr:fractional-x}, and } \nonumber
		\\\displaystyle&\sum_{S_i \not\in \A} x_i \cdot \min\{\deg_t(S_i, \A), k_t(\A)\} \ge k_t(\A), &\quad \forall \C_t \in \{\C_1, \ldots, \C_r\}  \text{ and } \forall \A \subseteq \R \label[constr]{constr:s-cover-ci-a}
		\end{alignat}
	\end{mdframed}
	
	
	Let $(x, z)$ be an \LP solution to the natural \LP\,---\, i.e., $(x, z)$ satisfies \Cref{constr:cover-ej,constr:cover-ci,constr:fractional-z,constr:fractional-x}. Let $H = \{e_j \in X \mid \sum_{ S_i \ni e_j} x_i \ge \frac{1}{6 \alpha} \}$ be the set of \emph{heavy} elements, where $\alpha > 1$ is some constant to be fixed later. Let $(\tilde{x})$ be a solution defined as $\tilde{x}_i = \min\{6\alpha \cdot x_i, 1\}$, for all $S_i \in \R$. By definition, for any heavy element $e_j \in H$, we have that $\sum_{S_i \ni e_j} \tilde{x}_i \ge 1$, so $(\tilde{x})$ is a feasible Set Cover \LP solution for the projected set system $(H, \R_{|H})$, with cost at most $6\alpha \cdot \sum_{S_i \in \R} x_i$. Let $\A'$ be the collection of sets returned by a $\beta$-approximate Set Cover \LP rounding algorithm. Starting from an empty solution, we add the sets from $\A'$ to the solution. Let $\A = \A' \cup \{S_i \in \R \mid x_i \ge \frac{1}{6\alpha} \}$. However, notice that if $x_i \ge \frac{1}{6\alpha}$ for some set $S_i$, then all elements in $S_i$ are heavy by definition, and hence are covered by $\A'$. Therefore, $\A$ and $\A'$ cover the same set of elements from $X$, and we may pretend that we have added the sets in $\A$ to our solution. 
	
	In \Cref{subsec:ellipsoid}, we discuss how to obtain a fractional \LP solution $(x, z)$ satisfying the following properties, in polynomial time. 
	\begin{enumerate}
		\item $(x, z)$ satisfies \Cref{constr:cover-ej,constr:cover-ci,constr:fractional-z,constr:fractional-x}.
		\item $\sum_{S_i \in \R} w_i x_i \le 2 \cdot OPT$
		\item $\sum_{S_i \not\in \A} x_i \cdot \min\{ \deg_{t}(S_i, \A), k_t(\A) \} \ge k_t(\A) \qquad \forall \C_t \in \{ \C_1, \ldots, \C_t \}$, where $\A$ is obtained from $(x, z)$ as described above.
	\end{enumerate}
	
	We assume that we have such a fractional \LP solution $(x, z)$. Now, for any uncovered element $e_j \in X \setminus \bigcup \A$, we have that $\sum_{ S_i \ni e_j} x_i < \frac{1}{6\alpha}$. Furthermore, for any set $S_i \in \R \setminus \A$, we have $x_i < \frac{1}{6\alpha}$. However, even after adding $\A$ to the solution, each color class $\C_t$ is left with some residual coverage $k_t(\A)$. In order to satisfy this residual coverage requirement, we use the following randomized rounding algorithm.
	
	Consider an iteration $\ell$ of the algorithm. In this iteration, for each set $S_i \in \R \setminus \A$, we independently add $S_i$ to the solution with probability $6x_i$. Let $\Sigma_\ell$ be the collection of sets added during this iteration in this manner. We perform $c \log r$ independent iterations of this algorithm for some constant $c$, and let $\Sigma = \bigcup_{\ell = 1}^{c \log r} \Sigma_\ell$. In \Cref{subsec:coverage}, we prove the following property about $\Sigma_\ell$.
	\begin{restatable}{lemma}{coverage}
		\label{lem:constant-prob}
		For any color class $\C_t$, the solution $\Sigma_\ell$ covers at least $k_t(\A)$ elements from $\C_t \setminus \C_t(\A)$ with probability at least a positive constant.
	\end{restatable}
	
	This lemma implies the following result.
	
	\begin{theorem} \label{thm:main-theorem}
		Suppose there exists a polynomial time \LP rounding algorithm that rounds a given Set Cover \LP solution on any projection of $(X, \R)$, within a $\beta$ factor. Then, there exists a polynomial time randomized algorithm that returns a solution with approximation guarantee of $O(\beta + \log r)$ for the Partition Set Cover problem, with at least a constant probability.
	\end{theorem}
	\begin{proof}
		First, we argue about the running time. As described earlier, the discussion of how to obtain the required \LP solution $(x, z)$ in polynomial time is deferred to \Cref{subsec:ellipsoid}. It is easy to see that the rest of the steps take polynomial time, and hence the overall time taken by this algorithm is polynomial.
		
		Now, we argue about the feasibility of this solution. The sets in $\A'$ form a cover for the heavy elements by assumption. Furthermore, from \Cref{lem:constant-prob}, in any iteration $\ell$, the sets in $\Sigma_\ell$ satisfy the remaining coverage requirement of any color class $\C_t$, with at least a constant probability. It follows from union bound that the probability that there exists a color class with unmet coverage requirement after $c \log r$ independent iterations is at most $1/(2r)$, for appropriately chosen constant $c$. Therefore, the solution $\A' \cup \Sigma$ is feasible with probability at least $1-1/(2r).$
		
		As argued earlier, $w(\A') \le 6\alpha \beta \cdot \sum_{S_i \in \R} w_i x_i = O(\beta) \cdot OPT$, since $\alpha$ is a constant. It is also easy to see that for any iteration $\ell$ of the randomized rounding, $\E[w(\Sigma_\ell)] \le 6 \cdot \sum_{S_i \in \R} w_i x_i = O(1) \cdot OPT$. Therefore, $\E[w(\Sigma)] \le c' \log r \cdot OPT$, for some constant $c'$. Therefore, $\Pr[w(\Sigma) \le 3c' \log r \cdot OPT] \ge \frac{2}{3}$, using Markov's inequality. The theorem follows from an application of union bound over the events concerning the feasibility and the approximation guarantee of the solution.
	\end{proof}

	\subsection{Analyzing Coverage of \texorpdfstring{$\Sigma_\ell$}{Σ\_l}} \label{subsec:coverage}
	As stated earlier, we show that for any color class $\C_t$, the probability that $\Sigma_\ell$ covers at least $k_t(\A)$ elements is at least a positive constant. Let $\A$ be the collection of sets as defined earlier. Recall that for any color class $\C_t$, the solution $(x, z)$ satisfies:
	$$\sum_{S_i \not\in \A} x_i \cdot \min\{\deg_t(S_i, \A), k_t(\A)\} \ge k_t(\A).$$
	
	Henceforth, fix a color class $\C_t$. To simplify notation, let us use the following shorthands: $\C \coloneqq \C_t \setminus \C_t(\A)$ is the set of uncovered elements from $\C_t$. $k \coloneqq k_t(\A)$ is the residual coverage requirement. For any set $S_i \in \R$, we restrict it to its projection on $\C_t \setminus \C_t(\A)$. Similarly, for a collection of sets $\R' \subseteq \R \setminus \A$, we restrict $\bigcup \R'$ to mean the set of ``uncovered elements'' from this color class $\C_t$, i.e., $\bigcup \R' = \bigcup_{S_i \in \R'} S_i$ (where each $S_i \in \R'$ is the projection of the original set, as in the previous sentence). Finally, let $\delta_i \coloneqq \frac{\min\{\deg_t(S_i, \A), k\}}{k}$. Notice that using this notation, the preceding constraint is equivalent to $$\sum_{S_i \not\in \A} \delta_i x_i \ge 1.$$
	
	For a set $S_i$, let $\hat{x}_i$ be an indicator random variable that denotes whether or not $S_i$ was added to $\Sigma_\ell$. It is easy to see that $\E[\hat{x}_i] = \Pr[S_i \text{ is added}] = 6x_i$. For an element $e_j \in \C$, let $Z_j \coloneqq \sum_{S_i \ni e_j} \hat{x}_i$ be a random variable that denotes the number of sets containing $e_j$ that are added to $\Sigma_\ell$. Notice that $e_j$ is covered by $\Sigma_\ell$ iff $Z_j \ge 1$. Note that, $\E[Z_j] = \sum_{ S_i \ni e_j} \E[\hat{x}_i] = \sum_{S_i \ni e_j} 6 x_i < 6 \cdot \frac{1}{6\alpha} = \frac{1}{\alpha}$.
	
	Let $Z \coloneqq \sum_{S_i \not\in \A} \delta_i \cdot \hat{x}_i $ be a random variable. Notice that 
	\begin{equation}
		Z = \sum_{S_i \not\in \A} \delta_i \cdot \hat{x}_i \le \frac{1}{k} \sum_{S_i \not\in \A} \hat{x}_i \cdot \deg_{t}(S_i, \A) = \frac{1}{k} \sum_{S_i \not\in \A} \sum_{e_j \in S_i} \hat{x}_i = \frac{1}{k} \sum_{e_j \in \C} Z_j\ . \label[ineq]{ineq:Z}
	\end{equation}
	
	Here, the second equality follows from the fact that $\deg_t(S_i, \A)$ is exactly the number of elements in $S_i$ (that are not covered by $\A$). 
	
	Now following \citet{bera2014approximation}, we show the following fact about $Z$.
	
	\begin{claim} \label{lem:Z-var}
		$\Pr[Z < 2] \le \frac{3}{8}$.
	\end{claim}
	\begin{proof}
		First, notice that $$\E[Z] = \sum_{S_i \not\in \A} \delta_i \cdot \E[\hat{x_i}] = 6 \cdot \sum_{S_i \not\in \A} \delta_i x_i \ge 6.$$
		
		Now, consider 
		\begin{align*}
			\Var[Z] = \sum_{S_i \not\in \A} \delta_i^2 \cdot \Var[\hat{x}_i] = \sum_{S_i \not\in \A} \delta_i^2 \cdot 6x_i (1-6x_i) \le 6 \sum_{S_i \not\in \A} \delta_i x_i\ . \tag{$\because$ $0 \le \delta_i \le 1$.}
		\end{align*}
		
		Using Chebyshev's inequality, 
		\begin{align*}
			\Pr[Z < 2] \le \Pr\left[\big|Z - \E[Z]\big| \ge \frac{2\E[Z]}{3} \right] \le \frac{9}{4} \cdot \frac{\Var[Z]}{\E[Z]^2} \le \frac{9}{4} \cdot \frac{6 \sum_{S_i \not\in \A} \delta_i x_i}{(6 \sum_{S_i \not\in \A} \delta_i x_i)^2}  \le \frac{3}{8}\ . \tag{$\because \sum_{S_i \not\in \A} \delta_i x_i \ge 1$.}
		\end{align*}
	\end{proof}
	
	For convenience, let us use the following notation for some events of interest:
	 
	\begin{align*}
		\K &\equiv \Sigma_\ell \text{ covers at most $k-1$ elements from $\C$}
		\\\M &\equiv Z < 2
	\end{align*}
	Recall that the objective is to show that $\Pr[\K]$ is upper bounded by a constant less than $1$. To this end, we first analyze $\Pr[\M|\K]$. 
	
	\begin{claim} \label{lem:p-given-q}
		$\Pr[\M|\K] \ge \frac{2}{5}$
	\end{claim}
	\begin{proof}
		For an element $e_j \in \C$, define an event $$\L_j \equiv Z_j \ge 1 \text{ (i.e., $e_j$ is covered)}.$$
		First, consider the following conditional expectation:
		\begin{align}
			\E[Z|\K] &\le \frac{1}{k}\sum_{e_j \in \C} \E[Z_j | \K] \tag{From \ref{ineq:Z}}
			\\&= \frac{1}{k} \sum_{e_j \in \C} \Pr[\bar{\L_j} | \K] \cdot \E[Z_j | \K \cap \bar{\L_j}] + \Pr[\L_j | \K] \cdot \E[Z_j | \K \cap \L_j] \nonumber
			\\&= \frac{1}{k} \sum_{e_j \in \C} \Pr[\L_j | \K] \cdot \E[Z_j | \K \cap \L_j ]\ . \tag{$\because \E[Z_j | \K \cap \bar{\L_j}] = 0.$}
		\end{align}
		
		In \Cref{subsec:expectation}, we show that the conditional expectation $E[Z_j | \K \cap \L_j]$ is upper bounded by $\frac{6}{5}$. Then, it follows that,
		\begin{align*}
		\E[Z|\K] &\le \frac{6}{5k} \sum_{e_j \in \C} \Pr[Z_j \ge 1 | \K] 
		\\&= \frac{6}{5k} \sum_{\R' \subseteq \R \setminus \A} \Pr[\Sigma_\ell = \R' | \K] \cdot \left|\bigcup \R'\right| \tag{Where we sum over collections $\R' \subseteq \R \setminus \A$ s.t. $|\bigcup \R'| \le k-1$}
		\\&\le \frac{6(k-1)}{5k} \sum_{\R' \subseteq \R \setminus \A} \Pr[\Sigma_\ell = \R' | \K]
		\\&< \frac{6}{5}\ .
		\end{align*}
		
		Now, using Markov's inequality, we have that 
		$$\Pr[\M|\K] = \Pr[Z < 2 | \K] \ge 1 - \frac{\E[Z|\K]}{2} \ge 1- \frac{3}{5} = \frac{2}{5}\ .$$
	\end{proof}
	
	We conclude with proving the main result of the section, which follows from \Cref{lem:Z-var} and \Cref{lem:p-given-q}.
	\coverage*
	\begin{proof}
		Consider 
		$$\Pr[\K] = \frac{\Pr[\K|\M] \cdot \Pr[\M]}{\Pr[\M|\K]} \le \frac{\Pr[\M]}{\Pr[\M|\K]} \le \frac{3/8}{2/5} = \frac{15}{16}\ .$$
		Therefore, $\Pr[\Sigma_\ell \text{ covers at least $k_t(\A)$ elements from }\C_t \setminus \C_t(\A) ] = 1 - \Pr[\K] \ge \frac{1}{16}\ .$
	\end{proof}

	\subsection{Analyzing \texorpdfstring{$\E[Z_j | \K \cap \L_j]$}{E[Z\_j | K ∩ L\_j]}} \label{subsec:expectation}
	In this section, we show that for any $e_j \in \C$, the conditional expectation $\E[Z_j | \K \cap \L_j]$ is bounded by $\frac{6}{5}$. Recall that $\L_j$ denotes the event $Z_j \ge 1$ (equivalently, $e_j$ is covered by $\Sigma_\ell$), and that $\K$ denotes the event that $\Sigma_\ell$ covers at most $k-1$ elements. For notational convenience, we shorten $\L_j$ to $\L$.
	
	Partition the sets $\R \setminus \A$ into disjoint collections $\R_1, \R_2$, where $\R_1$ consists of sets that do not contain $e_j$, and $\R_2$ consists of sets that contain $e_j$. Fix an arbitrary ordering $\sigma$ of sets in $\R \setminus \A$, where the sets in $\R_1$ appear before the sets in $\R_2$. We view the algorithm for choosing $\Sigma_\ell$, as considering the sets in $\R \setminus \A$ according to this ordering $\sigma$, and making a random decision of whether to add each set in $\Sigma_\ell$. Let $\Sigma'_\ell$ be the random collection of sets added according to the ordering $\sigma$, until the first set containing $e_j$, say $S_i$, is added to $\Sigma_\ell$. Note that if we condition on the event $\L$, such an $S_i$ must exist. 
	
	Let \hi denote the event that (i) $S_i$ is the first set containing $e_j$ that is added by the algorithm, and (ii) $\Sigma'_\ell$ is the collection added by the algorithm, just after $S_i$ was added. Note that \hi contains the history of the choices made by the algorithm, until the point just after $S_i$ is considered. For a history \hi, $k-1-|\bigcup \Sigma'_\ell|$ is the maximum number of additional elements that can still be covered, without violating the condition $\K$ (which says that $\Sigma_\ell$ covers at most $k-1$ elements). We say that a history \hi is relevant, if $k-1-|\bigcup \Sigma'_\ell| \ge 0$. Thus,
	\begin{equation} \label{eqn:exp-1}
		\E[Z_j \mid \K \cap \L] = \sum_{\hi} \Pr[\hi \mid \K \cap \L] \cdot \E[Z_j \mid \K \cap \L \cap \hi],
	\end{equation} 
	where we only sum over the relevant histories.
	
	Now, once $S_i$ has been added to the solution, $e_j$ is covered, thereby satisfying the condition $\L$. That is, the event \hi implies the event $\L$. It follows that,
	\begin{equation} \label{eqn:exp-2}
		\E[Z_j \mid \K \cap \L \cap \hi] = \E[Z_j \mid \K \cap \hi].
	\end{equation}
	Let $\K'$ denote the event that $\Sigma_\ell \setminus \Sigma'_\ell$ covers at most $p \coloneqq k-1-|\bigcup \Sigma'_\ell|$ elements. Then,
	\begin{equation} \label{eqn:exp-3}
		\E[Z_j \mid \K \cap \hi] = \E[Z_j \mid \K' \cap \hi].
	\end{equation}
	Now, let $\bar{Z_j}$ be the sum of the indicator random variables $\hat{x}_{i'}$, over the the sets $S_{i'} \in \R_2$, that occur after $S_i$ in the ordering $\sigma$. Clearly, $\E[\bar{Z_j}] \le \E[Z_j]$. We also have,
	\begin{equation} \label{eqn:exp-4}
		\E[Z_j \mid \K' \cap \hi] = 1 + \E[\bar{Z_j} \mid \K' \cap \hi].
	\end{equation}
	This is because, $Z_j$ denotes the number of sets containing $e_j$ that are added to $\Sigma_\ell$, including $S_i$; whereas $\bar{Z_j}$ does not count $S_i$. Now, $\bar{Z_j}$ and $\K'$ are concerned with the sets after $S_i$ according to $\sigma$, whereas \hi concerns the history upto $S_i$. Therefore,
	\begin{align} 
		\E[\bar{Z_j} \mid \K' \cap \hi] &= \E[\bar{Z_j} \mid \K'] \nonumber
		\\&= \E[\bar{Z_j} |\Sigma_\ell \setminus \Sigma'_\ell \text{ covers at most } p \text{ additional elements}] \nonumber
		\\&\le \frac{\E[\bar{Z_j}]}{\Pr[\Sigma_\ell \setminus \Sigma'_\ell \text{ covers at most }p \text{ additional elements}]} \nonumber
		\\&\le \frac{\E[\bar{Z_j]}}{\Pr[\Sigma_\ell \setminus \Sigma'_\ell = \emptyset]} \tag{$\because$ ``$\Sigma_\ell \setminus \Sigma'_\ell$ covers at most $p$ additional elements'' $\supseteq$ ``$\Sigma_\ell \setminus \Sigma'_\ell = \emptyset$'' } \nonumber
		\\&\le \frac{1}{\alpha - 1}\ . \label{eqn:exp-5}
	\end{align}
	This follows from (i) $\E[\bar{Z_j}] \le \E[Z_j] \le \frac{1}{\alpha}$ as argued earlier, and (ii) $\Pr[\Sigma_\ell \setminus \Sigma'_\ell = \emptyset] \ge \prod_{S_{i'} \ni e_j }(1 - 6x_{i'}) \ge 1 - \sum_{S_{i'} \ni e_j} 6x_{i'} \ge \frac{\alpha - 1}{\alpha}$, where we use Weierstrass product inequality in the second step. 
	
	Now, combining \Cref{eqn:exp-2,eqn:exp-3,eqn:exp-4,eqn:exp-5}, we conclude that
	$$\E[Z_j \mid \K \cap \L \cap \hi] \le 1 + \frac{1}{\alpha-1} = \frac{\alpha}{\alpha -1}.$$
	Plugging this into \Cref{eqn:exp-1}, we get that, 
	$$\E[Z_j \mid \L\cap \K] \le \sum_{\hi} \frac{\alpha}{\alpha - 1} \cdot  \Pr[\hi \mid \K\cap \L] = \frac{\alpha}{\alpha - 1} \sum_{\hi} \Pr[\hi \mid \K\cap \L] = \frac{\alpha}{\alpha-1}.$$

	Choosing $\alpha = 6$, it follows that $E[Z_j | \L\cap \K] \le \frac{6}{5}$, as claimed.
	
	\subsection{Solving the \LP} \label{subsec:ellipsoid}
	Recall that the strengthened \LP has exponentially many constraints, and hence we cannot use a standard \LP algorithm directly. We guess the cost of the integral optimal solution up to a factor of $2$ (this can be done by binary search), say $\Delta$. Then, we convert the \LP into a feasibility \LP by removing the objective function, and adding a constraint $\sum_{S_i \in \R} w_i x_i \le \Delta$. We then use Ellipsoid algorithm to find a feasible solution to this \LP and let $(x, z)$ be a candidate solution returned by the Ellipsoid algorithm. If it does not satisfy the preceding constraint, we report it as a violated constraint. We also check \Cref{constr:cover-ej,constr:cover-ci,constr:fractional-z,constr:fractional-x} (the number of these constraints is polynomial in the input size), and report if any of these constraints is violated. 
	
	Otherwise, let $H$ be the set of heavy elements with respect to $(x, z)$ (as defined earlier), and let $\A = \A' \cup \{S_i \in \R \mid x_i \ge \frac{1}{6\alpha}\}$ be the collection of sets as defined in \Cref{sec:LPRounding} (where $\A'$ is the Set Cover solution for the heavy elements $H$ returned by the rounding algorithm). Then, we check if the following constraint is satisfied with respect to this $\A$, for all color classes $\C_t$:
	$$\sum_{S_i \not\in \A} x_i \min\{ \deg_{t}(S_i, \A), k_t(\A) \} \ge k_t(\A)\ .$$
	 
	If this constraint is not satisfied for some color class $\C_t$, we report it as a violated constraint. Otherwise we stop the Ellipsoid algorithm and proceed with the randomized rounding algorithm with the current \LP solution $(x, z)$, as described in \Cref{sec:LPRounding}. Note that $(x, z)$ may not satisfy \Cref{constr:s-cover-ci-a} with respect to all collections $\A$. However, our randomized rounding algorithm requires that it is satisfied with respect to the specific collection $\A$ as defined earlier. Therefore, we do not need to check the feasibility of $(x, z)$ with respect to exponentially many constraints.
	
	\section{Facility Location and Minimum Cost Covering with Multiple Outliers}
	\label{sec:fl-mcc}
	
	We consider generalizations of the Facility Location and Minimum Cost Covering problems. These generalizations are analogous to the Partition Set Cover problem considered in the previous section. That is, the set of ``clients'' (which are the objects to be covered) is partitioned into $r$ color classes, and each color class has a coverage requirement. We note that for the (standard) Facility Location and Minimum Cost Covering problems, \LP-based $O(1)$ approximation algorithms are known. In the following, we first state the generalizations formally, and then show how the Randomized Rounding framework from the previous section can be adapted to obtain $O(\log r)$ approximations for these problems. These guarantees are asymptotically tight, in light of the hardness results given in \Cref{sec:hardness}.
	
	 \subsection{Facility Location with Multiple Outliers}
	
	In the Facility Location with Multiple Outliers problem, we are given a set of Facilities $F$, a set of Clients $C$, belonging to a metric space $(F \cup C, d)$. Each facility $i \in F$ has a non-negative opening cost $f_i$. We are given $r$ non-empty subsets of clients (or ``color classes'') $\C_1, \ldots, \C_r$, that partition the set of clients. Each color class $\C_t$ has a connection requirement $1 \le k_t \le |\C_t|$. The objective of the Facility Location with Multiple Outliers problem is to find a solution $(F^*, C^*)$, such that $\sum_{i \in F'} f_i + \sum_{j \in C'} d(j, F')$ is minimized over all feasible solutions $(F', C')$. A solution $(F', C')$ is feasible if (i) $|F'| \ge 1$ and (ii) For all color classes $\C_t$, $|\C_t \cap C'| \ge k_t$. Note that this is a generalization of the Robust Facility Location problem, first considered by \citet{Charikar2001FLwO}. 
	
	A natural \LP formulation of this problem is as follows. 
	
	\begin{mdframed}[backgroundcolor=gray!9] 
		(Natural \LP for Facility Location with Multiple Outliers)
		\begin{alignat}{3}
		\text{minimize}   &\ \sum\limits_{i \in F} f_{i}x_{i} + \sum_{i \in F, j \in C} y_{ij} \cdot d(i, j) & \nonumber \\
		\text{subject to} \displaystyle&\sum\limits_{i \in F}   y_{ij} \geq z_j,  \quad &  \forall j \in C \label[constr]{constr:fl-cover-j}\\
		\displaystyle&\sum_{j \in \C_r}z_j \ge k_t, & \forall \C_t \in \{\C_1, \ldots, \C_r\} \label[constr]{constr:fl-cover-ci}\\
		\displaystyle&0 \le  y_{ij} \le x_i \le 1, & \forall i \in F, \forall j \in C \label[constr]{constr:fl-atmost-ci}\\
		\displaystyle &z_j \in [0, 1], &  \forall j \in C \label[constr]{constr:fl-fractional-z}
		\end{alignat}
	\end{mdframed}
	
	We note that the integrality gap example from \Cref{subsec:nat-lp} can be easily converted to show a similar gap for the Facility Location with Multiple Outliers problem. Therefore, we strengthen the \LP in a manner similar to the previous section.
	
	First, we convert the \LP to a feasibility \LP by guessing the optimal cost up to a factor of $2$, say $\Delta$, and by adding a constraint $\sum\limits_{i \in F} f_{i}x_{i} + \sum_{i \in F, j \in C} y_{ij} \cdot d(i, j) \le \Delta$. Similar to \Cref{subsec:ellipsoid}, we use the Ellipsoid algorithm to find a feasible \LP solution that satisfies this constraint, as well \Crefrange{constr:fl-cover-j}{constr:fl-fractional-z}. Let $H = \{j \in C \mid \sum_{i \in F} y_{ij} \ge \frac{1}{6\alpha} \}$ be the set of \emph{heavy} clients. For any $i \in F$, let $\tilde{x}_i \coloneqq \min \{1, 6\alpha \cdot x_i \}$, and for any $i \in F, j \in H$, let $y_{ij} \coloneqq \min\{1, 6\alpha \cdot y_{ij}\}$. It is easy to see that $(\tilde x, \tilde y)$ is a feasible Facility Location (without outliers) solution for the instance induced by the heavy clients, and its cost is at most $6\alpha \Delta$. We use an \LP-based algorithm (such as \cite{ByrkaFLLP}) with a constant approximation guarantee to round this solution to an integral solution $(F_H, H)$, where $F_H \subseteq F$. 
	
	Let $L = C \setminus H$ be the set of \emph{light} clients. Note that for any light client $j \in L$, $z_j \le \sum_{i \in F} y_{ij} < \frac{1}{6\alpha}$. Also, for a color class $\C_t$, let $\C_t(H) \coloneqq \C_t \setminus H$ denote the uncovered (light) elements from $\C_t$, and let  $k_t(H) \coloneqq k_t - |\C_t \cap H|$ denote its residual coverage requirement. Wlog, we assume that $k_t(H)$ is positive, otherwise we can ignore the color class $\C_t$ from consideration in the remaining part. Now, we check whether the following constraint holds for all color classes $\C_t$:
	
	\begin{equation}
	\sum_{i \in F} \min\bigg\{x_i \cdot k_t(H), \sum_{j \in \C_t(H)} y_{ij}  \bigg\} \ge k_t(H) \label[constr]{eqn:fl-constraint}
	\end{equation}
	
	 First, note that this can be easily formulated as an \LP constraint by introducing auxiliary variables. If this constraint is not satisfied for some color class $\C_t$, we report it as a violated constraint. Consider the integral \LP solution $(x', y', z')$ corresponding to a feasible integral solution $(F', C')$. We argue that $(x', y', z')$ satisfies this constraint. Note that at most $|\C_t \cap H|$ clients are connected from the set $\C_t \cap H$. Therefore, by feasibility of the solution, at least $k_t(H)$ clients must be connected from $\C_t(H)$. For a facility $i \in F'$, the quantity $\sum_{j \in \C_t(H)} y'_{ij}$ denotes the number of clients connected to $i$. However, even if the number of clients connected to $i$ is more than $k_t(H)$, only $k_t(H)$ of them count towards satisfying the residual connection requirement. Therefore, $(x', y', z')$ satisfies this constraint for all color classes, and hence it is a valid constraint.
	
	Now, suppose we have an \LP solution $(x, y, z)$ that satisfies \Crefrange{constr:fl-cover-j}{eqn:fl-constraint}, and has cost at most $\Delta$. By ``splitting'' the facilities into multiple co-located copies if necessary, we ensure the following two conditions hold:
	\begin{enumerate}
		\item For any facility $i \in F$, $x_i < \frac{1}{6\alpha}$.
		\item For any client $j \in L$ and any facility $i \in F$, $y_{ij} > 0\ \implies\ y_{ij} = x_i$.
	\end{enumerate}
	This has to be done in a careful manner, since we also want to maintain \Cref{eqn:fl-constraint} after the facilities have been split. This procedure results in a feasible \LP solution of the same cost. Henceforth, we treat all co-located copies of a facility as distinct facilities for the sake of the analysis. We now show that the rounding for the light clients can be reduced to the Randomized Rounding algorithm from the previous section.
	
	For any facility $i \in F$, let $S_i \coloneqq \{j \in L \mid x_{i} = y_{ij} \}$ denote the set of light clients that are fractionally connected to $i$. The cost of opening facility $i$ and connecting all $j \in S_i$ to $i$ is equal to $w_i \coloneqq f_i + \sum_{j \in S_i} d(i, j)$. Consider an instance $(L, \R)$ of the Partition Set Cover problem, where $\R = \{S_i \mid i \in F\}$ with weights $w_i$, and residual coverage requirement $k_t(H)$ for each color class $\C_t(H)$, and consider the corresponding \LP solution $(x, z)$. The following properties are satisfied by the \LP solution. 
	\begin{enumerate}
		\item All the elements are light, and all the sets $S_i \in \R$ have $x_i < \frac{1}{6\alpha}$.
		\item The costs of the two \LP solutions are equal: $$\sum_{S_i \in \R} w_i x_i = \sum_{i \in F} x_i \cdot \bigg(f_i + \sum_{j \in S_i} d(i, j)\bigg) = \sum_{i \in F} f_i x_i + \sum_{i \in F,\ j \in L} y_{ij} \cdot d(i, j).$$
		\item \Cref{eqn:fl-constraint} is equivalent to:
		$$\sum_{S_i \in \R} x_i \cdot \min\left\{ k_t(H), |S_i \cap \C_t| \right\} \ge k_t(H) \quad \forall \C_t.$$
	\end{enumerate}
	
	Therefore, we can use the Randomized Rounding algorithm from the previous section to obtain a solution $\Sigma = \bigcup_{\ell = 1}^{O(\log r)} \Sigma_\ell$. It has cost at most $O(\log r) \cdot \Delta$, and for each color class $\C_t(H)$, it covers at least $k_t(H)$ clients, with at least a constant probability. To obtain a solution for the Facility Location with Multiple Outliers problem, we open any facility $i \in F$, if its corresponding set $S_i$ is selected in $\Sigma$. Furthermore, we connect $k_t(H)$ clients from $\C_t(H)$ to the set of opened facilities. Note that the cost of this solution is upper bounded by $w(\Sigma) \le O(\log r) \cdot \Delta$. Combining this with the solution $(F_H, H)$ for the heavy clients with cost at most $O(1) \cdot \Delta$, we obtain our overall solution for the given instance. It is easy to see that this is an $O(\log r)$ approximation.

	\subsection{Minimum Cost Covering with Multiple Outliers}
	
	Here, we are given a set of Facilities $F$, a set of Clients $C$, belonging to a metric space $(F \cup C, d)$. Each facility $i \in F$ has a non-negative opening cost $f_i$. We are given $r$ subsets of clients (or ``color classes'') $\C_1, \ldots, \C_r$, where any client $j \in C$ belongs to at least one color class. Each color class $\C_t$ has a coverage requirement $1 \le k_t \le |\C_t|$. A ball centered at a facility $i \in F$ of radius $r \ge 0$ is the set $B(i, r) \coloneqq \{j \in C \mid d(i, j) \le r \}$. The goal is to select a set of balls $\B = \{B_i = B(i, r_i) \mid i \in F' \subseteq F \}$ centered at some subset of facilities $F' \subseteq F$, such that (i) The set of balls $\B$ satisfies the coverage requirement of each color class and (ii) the sum $\sum_{i \in F'} (f_i + r_i^\gamma)$ is minimized. Here, $\gamma \ge 1$ is a constant, and is a parameter of the problem.
	
	Note that even though the radius of a ball centered at $i \in F$ is allowed to be any non-negative real number, it can be restricted to the following set of ``relevant'' radii: $R_i \coloneqq \{d(i, j)\mid j \in C\}$. Now, define a set system $(C, \R)$. Here, $C$ is the set of clients, and $\R = \{ B(i, r) \mid i \in F, r \in R_i \}$, with weight of the set corresponding to a ball $B(i, r)$ being defined as $f_i + r^\gamma$. Now, we use the algorithm from the previous section for this set system. Let $H$ be the set of heavy clients (or elements) as defined in \Cref{sec:LPRounding}. We use the Primal-Dual algorithm of \citet{CharikarP04} \footnote{\citet{CharikarP04} consider the special case of $\gamma = 1$, however their algorithm easily generalizes to arbitrary $\gamma$.} with an approximation guarantee of $\beta = 3^\gamma$ (which is a constant) to obtain a cover for the heavy clients. For the remaining light clients, we use the Randomized Rounding algorithm as is. Note that this reduction from the Minimum Cost Covering with Multiple Outliers Problem to the Partition Set Cover Problem is not exact, since the solution thus obtained may select sets corresponding to concentric balls in the original instance. However, from each set of concentric balls, we can choose the largest radius ball. This pruning process does not affect the coverage, and can only decrease the cost of the solution. Therefore, it is easy to see that the resulting solution is an $O(\log r)$ approximation.
	
	\section{\texorpdfstring{$\Omega(\log r)$}{Ω(log r)} Hardness Results} \label{sec:hardness}
	
	In this section, we show that it is \NP-hard to obtain approximation guarantees better than $O(\log r)$ for the Partition Set Cover for several geometric set systems, as well as the problems considered in \Cref{sec:fl-mcc}. The reductions are from (unweighted) Set Cover, and are straightforward extensions of a similar hardness result shown in \cite{bera2014approximation}.
	
	\subsection*{Geometric Set Systems}
	
	Suppose we are given an instance of Set Cover $(X, \R)$, where $X = \{e_1, \ldots, e_n\}$. For each set $S_i \in \R$, add a unit interval $I_i$ in $\real$, such that all intervals are disjoint. Add a point $p_{ij}$ (of color class $\C_j$) inside an interval $I_i$, corresponding to an element $e_j \in X$, and a set $S_i \ni e_j$. Thus, there are $n$ disjoint color classes, partitioning the set of points. The coverage requirement of each color class is $1$. It is easy to see that a feasible solution to the Partition Set Cover instance corresponds to a feasible solution to the original Set Cover instance, of the same cost. The $\Omega(\log r)$ hardness follows from the $\Omega(\log n)$ hardness for Set Cover (\cite{DS2014}), and $r = n$ is the number of color classes. Therefore, $\Omega(\log r)$ hardness follows, for the Partition Set Cover problem where the sets are unit intervals in $\real$. This is easily generalized to any other type of geometric objects, as all that is needed is the disjointness of the geometric objects.
	
	\subsection*{Facility Location and Minimum Cost Covering}
	
	First, consider the Facility Location with Multiple Outliers problem. Given an instance $(X, \R)$ of the unweighted Set Cover problem, we add facilities $i \in F$, corresponding to sets $S_i \in R$, uniformly separated on the real line, such that the distance between the facilities is at least $|X| \cdot |\R|$. The opening cost of each facility is $1$. Similar to the reduction above, we add a client $c_{ij}$ co-located with facility $i \in F$, corresponding to an element $e_j \in X$, and a set  $S_i \ni e_j$. The coverage requirement of each color class is set to $1$. It is easy to see the one-to-one correspondence between optimal solutions to both of these problems. 
	
	Now, we tweak the above instance for obtaining the same result for the Minimum Cost Covering with Multiple Outliers problem on a line, even when all the opening costs are $0$ (otherwise, we can use the reduction from the paragraph above as is). For each facility $i \in F$, add the clients $c_{ij}$ (corresponding to the elements in $S_i$) at a distance of $1$ from $i$ (instead of being co-located, as in the previous reduction). The coverage requirement of each color class is $1$, as before. Note that the facility $i$ and the clients $\{c_{ij} \mid e_j \in S_i \}$ form a ``cluster'', and the inter-cluster distance is large enough to ensure that an optimal solution to the resulting instance consists of disjoint clusters, which then exactly corresponds to an optimal solution to the Set Cover problem.
	
	\subsection*{Acknowledgment}
	We thank Sariel Har-Peled and Timothy M. Chan for preliminary discussions on this problem.
	\bibliography{references}
\end{document}